\documentclass[12pt]{article}

\usepackage{amsfonts}
\usepackage{amssymb}
\usepackage{amsmath}
\usepackage{amsthm}
\usepackage{tikz}
\usepackage{graphicx}
\usepackage{multirow}
\usepackage{subfigure}
\usepackage{eurosym}
\usepackage{xcolor}
\usepackage{enumitem}
\usepackage{pifont}

\RequirePackage[colorlinks,citecolor=blue,linkcolor=red,urlcolor=blue]{hyperref}
\RequirePackage{natbib}

\usetikzlibrary{shapes,decorations}

\newtheorem{theorem}{Theorem}
\newtheorem{proposition}[theorem]{Proposition}

\newtheorem{lemma}[theorem]{Lemma}
\newtheorem{example}[theorem]{Example}
\newtheorem{definition}[theorem]{Definition}

\newtheorem{axiom}{Axiom}
\newtheorem{nsc}{No Shifted Cycles (NSC)}

\DeclareFontFamily{U}{mathx}{\hyphenchar\font45}
\DeclareFontShape{U}{mathx}{m}{n}{
	<5> <6> <7> <8> <9> <10>
	<10.95> <12> <14.4> <17.28> <20.74> <24.88>
	mathx10
}{}
\DeclareSymbolFont{mathx}{U}{mathx}{m}{n}
\DeclareMathSymbol{\bigtimes}{1}{mathx}{"91}

\makeatletter

\renewcommand*{\@seccntformat}[1]{%
	\csname the#1\endcsname.\quad}
\makeatother

\begin{document}
	
	\title{Choice via AI}
	
	\author{
		\textsc{Christopher Kops}\footnote{Department of Quantitative Economics, Maastricht University, P.O. Box 616, 6200 MD, Maastricht, The Netherlands; Homepage: \url{https://sites.google.com/site/janchristopherkops/home}; E-mail: \href{mailto: j.kops@maastrichtuniversity.nl}{\texttt{ j.kops@maastrichtuniversity.nl}}} \\
		\small{\textit{Maastricht University}}\\ \\
		\textsc{Elias Tsakas}\footnote{Department of Microeconomics and Public Economics, Maastricht University, P.O. Box 616, 6200 MD, Maastricht, The Netherlands; Homepage: \url{www.elias-tsakas.com}; E-mail: \href{mailto:e.tsakas@maastrichtuniversity.nl}{\texttt{e.tsakas@maastrichtuniversity.nl}}} \\
		\small{\textit{Maastricht University}}}
	
\date{\small{\today}}

	\maketitle
	
\begin{abstract}
\noindent This paper proposes a model of choice via agentic artificial intelligence (AI). A key feature is that the AI may misinterpret a menu before recommending what to choose. A single acyclicity condition guarantees that there is a monotonic interpretation and a strict preference relation that together rationalize the AI's recommendations. Since this preference is in general not unique, there is no safeguard against it misaligning with that of a decision maker. What enables the verification of such AI alignment is interpretations satisfying double monotonicity. Indeed, double monotonicity ensures full identifiability and internal consistency. But, an additional idempotence property is required to guarantee that recommendations are fully rational and remain grounded within the original feasible set.
\vspace{0.3\baselineskip}
		
		
\end{abstract}

\section{Introduction}

Suppose that a decision maker (DM) seeks advice from an AI Agent before choosing from a menu of available choices. Based on the recommendations that she receives, she is interested in answering two questions. First, does the AI make rational recommendations? Second, are these recommendations aligned with her own preferences?

In this paper, we address these questions in the context of a choice theoretic framework. In particular, we view the AI as a choice function that makes a recommendation for each menu. The only caveat in comparison with the standard framework is that the AI may misinterpret the menu from which the DM is seeking to choose. As a result, it may ignore some of the choices in the menu (similarly to the literature on consideration sets), or it may even end up making a recommendation that does not belong to the menu. For instance, suppose that the DM asks the AI to choose ``the most interesting paper from the Economics literature on political polarization". However, while the DM has a clear definition of the Economics literature (e.g., the set of Economics journals), the AI has a broader interpretation and thus ends up recommending a paper from the Annual Review of Political Science. 

Such misinterpretations are common in AI's that are based on large language models (LLM's), whose interpretation of a context typically depends on the data on which it has been trained. But more importantly, the AI's interpretation of each context is unobservable to the DM. Throughout the paper, we will only impose a weak monotonicity condition on the interpretation operator, according to which basic implications are interpreted correctly, i.e., the AI understands that a paper belonging to the Economics literature also belongs to the literature on Social Sciences, even if it distorts the meaning of ``Economics literature" and ``Social Sciences literature" respectively. 

Going back to our first question, the problem boils down to identifying violations of the usual rationality postulates despite the fact that these are possibly confounded with distortions due to potential misinterpretations of the menus. In particular, we ask whether there is a monotonic interpretation operator and a strict preference relation that together rationalize the observed recommendations. 

Our first main result establishes that a single condition in the style of a classic acyclicity condition allows us characterizes the set of rationalizable choice functions. This is a major step in that it guarantees that potential mistakes made by the AI can be traced back solely to the misinterpretations of the choice problems, rather than on choice inconsistencies. As such, training the AI with more data so that it learns to interpret the choice environment correctly, paves the way to fully solve the underlying problem.

Yet, preferences may remain only partially identifiable. Even though we can tell whether the AI's recommendations are rationalizable, we cannot uniquely identify the preference relation, as recommendations are still distorted by the AI's misinterpretations. Interestingly, preference identification does not guarantee correct interpretation of the choice problems. Instead, the two components of our model (viz., the preference relation and the interpretation operator) can be identified jointly from the choice function if the interpretation operator satisfies a double monotonicity condition. According to this condition, not only do we need the AI to correctly interpret logical implications, but moreover whenever the AI thinks that A implies B, this implication does actually hold. In the context of our earlier example, if the AI interprets the Economics literature as a subset of the Social Sciences literature, it is right to do so.

Our second main result then characterizes the set of choice functions that are rationalizable by a strict preference relation together with a doubly monotonic interpretation operator. This is done, by means of standard conditions which rule out binary cycles and choice reversals. 

The latter is enough in order to test whether the AI's recommendations are aligned with the DM's preferences. Yet, this does not mean that the AI's recommendations are the ones an agent with correct interpretation of the world would have made. To achieve the latter, we would need to guarantee that the AI's interpretation operator is grounded, i.e., it interprets choice problems correctly.

Our third main result characterizes the choice functions that are consistent with a strict preference relation and a grounded interpretation operator by means of the usual revealed axiom of revealed preference (WARP), thereby making the AI's recommendations rational in the traditional economic sense of the weak axiom of revealed preferences.

While our paper has been framed in the context of recommendations received by an AI agent ---which we view as a relevant and timely application--- all our results hold verbatim in any setting where the DM receives choice recommendations by an advisor whose interpretation of the world might be distorted and/or the preferences might differ from the ones of the DM.

This paper is related to the broad literature on choice behavior when only a subset of the available alternatives is considered for choice, because of factors such as norms \citep{Baigent1996}, framing effects \citep{Salant2008}, categorization \citep{Manzini2012a}, clustering \citep{Kops2022}, cognitive limitations \citep{Eliaz2011,Masatlioglu2012,Lleras2017} or social influence \citep{Cuhadaroglu2017,Borah2020}. The main difference of our work is that we also allow for consideration of alternatives that are not in the actual choice set. 

The behavioral foundation of our choice procedure also shares similarities to the literature modeling choice as the result of (sequentially) maximizing more than one preference relation \citep{Manzini2007,Manzini2012b, Cherepanov2013,Tyson2015,Horan2015,Horan2016,Kops2018,Armouti-Hansen2018} or maximizing over tuples of alternative and opportunity costs \citep{Kops2024,Manzini2025}. At a broader level, this paper also relates to applications of choice theoretic models to competitive markets \citep{Eliaz2011b}, finance \citep{Stango2014} and medical decision-making \citep{Katz2019}.

Our work also ties well with some earlier work on AI. Since we assume an AI agent rather than a human being as the central decision-maker, our model has to incorporate how LLMs transform \citep{Vaswani2017} the data they are fed. An immediate consequence lies in the so-called alignment problem of whether or not the AI's preferences clash with the preferences of a human decision-makers. In this regard, \cite{Kim2024} check whether LLMs are able to learn preferences and provide recommendations when it comes to choice under risk. \cite{Caplin2025} test how well machine learning predictions align with cognitive economic model.

The paper is organized as follows. Section \ref{SECT:Setup} lays out the setup. Section \ref{SECT:AI Agent} defines the model of choice by an agentic AI. Section \ref{SECT:Behavioral Foundation} provides characterization and identification. Section \ref{SECT:Rational AI Agent Choice} shows how to extend the baseline model to guarantee internal consistency. Section \ref{SECT:Grounded and Rational AI Agent Choice} develops it even further to also ensure groundedness.  Section \ref{SECT:Elaborations and Extensions} elaborates on some of our modeling assumptions and presents an extension guaranteeing groundedness, but not rationality. Section \ref{SECT:Conclusion} concludes.

\section{Setup}\label{SECT:Setup}
Let $X$ be a finite set of alternatives with typical elements denoted by $x$, $y$, $z$ etc. $\mathcal{P}(X)$ denotes the set of non-empty subsets of $X$ with typical elements $R$, $S$, $T$ etc. We refer to any such set as a choice problem. A choice function on $X$ is a mapping, $c:\mathcal{P}(X)\to X$.




\section{AI Agent}\label{SECT:AI Agent}
Unlike LLMs which only generate text, an AI agent can translate intent into procedures carried out in the real world. An AI agent can act and make choices. In particular, it can provide recommendations to a DM who seeks its advice.

Modern AI agents have an LLM with a transformer architecture \citep{Vaswani2017} as their policy core. This architecture enables them to convert text into tokens, but also to further contextualize and interpret them. The AI's interpretation of a context then typically depends on the data on which it has been trained. We assume that training and fine-tuning of the LLM enables the AI to correctly interpret basic logical implications, i.e., the AI understands the subset relations in $\mathcal{P}(X)$.

\begin{definition}\label{DEF:Interpretation operator}\textsc{(Interpretation Operator).}
An interpretation operator is a mapping $I:\mathcal{P}(X)\rightarrow \mathcal{P}(X)$ that satisfies the following basic condition:
\begin{itemize}[leftmargin=30pt]
\item[$(I_M)$] \textsc{Monotonicity ---} For all $S,T\in \mathcal{P}(X):$ $S\subseteq T \Rightarrow {I}(S)\subseteq {I}(T)$
\end{itemize}
\end{definition}

Monotonicity is the single key property of the interpretation operator. It ensures that howsoever the AI distorts the meaning of certain choice problems, these distortions still preserve the subset-relationship between the non-empty subsets of $X$.\footnote{In the context of artificial intelligence \citep{Russell1995}, especially when it comes to heuristic functions, monotonicity is also often referred to as consistency.} This is the minimal requirement that training and fine-tuning the AI should achieve.

We can now introduce the choice procedure that we are proposing in this paper. The key feature of this procedure is that the AI agent chooses on behalf of our rational DM. So, the procedure requires that the AI picks the best element from the interpreted choice problem according to a preference relation.

\begin{definition}
A choice function $c:\mathcal{P}(X)\to X$ is an AI agent's choice (AIC) if there exists a strict preference relation $\succ$ on $X$ and an interpretation operator $I$ on $\mathcal{P}(X)$ such that for any choice problem $S\in\mathcal{P}(X)$,
$$
c(S)=\{x\in {I}(S)|\; x\succ y,\text{ for all }y\in {I}(S)\setminus\{x\}\}
$$
\end{definition}

Under the AIC, $I(S)$ can also be thought of as resulting from a Retrieval-Augmented Generation (RAG) step \citep{Shi2024}, where the AI retrieves a context before generating a recommendation.

\section{Behavioral Foundation}\label{SECT:Behavioral Foundation}
This section provides a behavioral (i.e., choice-based) characterization of the AIC procedure. It identifies which choice functions are rationalizable by interpretation operator and preference relation.

Let us reiterate here that such a characterization enables any outside observer to verify whether choice data is consistent with the AIC procedure or not. As it turns out, this procedure can be characterized by the following single axiom.

\begin{nsc}
For all $S_1,\dots,S_{n+1},T_1,\dots,T_{n+1}\in \mathcal{P}(X)$
$$
[c(S_i)=x_i,\; S_i\subset T_i,\forall i=1,\dots,n+1,\text{ and }c(T_i)=x_{i+1},\forall i=1,\dots,n] \Rightarrow [c(T_{n+1})\neq x_1]
$$
\end{nsc}

NSC is in the style of a classic acyclicity condition such as No Binary Cycles \citep{Manzini2025}. The AI's choice from the $T_i$'s is all about this acyclicity. Indeed, $c(T_1)=x_2$, $c(T_2)=x_3$, \dots, $c(T_{n})=x_{n+1}$, $c(T_{n+1})=x_{1}$ would reveal preferences to be cyclic.\footnote{Proposition \ref{PROP:revealed preference} and its ensuing paragraph illustrate how NSC can be restated directly in terms of the AI's revealed preferences.} But, the definition of AIC guarantees that the AI's preferences are acyclic. That is why NSC restricts the AI's choice to satisfy $c(T_{n+1})\neq x_1$. The added layer of complexity comes from the AI interpreting choice problems. The AI's choice from the $S_i$'s then establishes that each choice problem $T_{i}$ is interpreted in such a way that the interpreted choice problem ${I}(T_i)$ involves $x_i$.

The following result then establishes that NSC provides a behavioral characterization of AIC.

\begin{theorem}\label{THM:main result}
Let $X$ be a finite set of alternatives. A choice function $c$ on $X$ is an AIC if and only if it satisfies NSC.
\end{theorem}
\begin{proof}
Please refer to Appendix \ref{PROOF:main result}.
\end{proof}

The AIC is based on two key parameters that enter into the AI's decision-making procedure: the strict preference ranking $\succ$ and the interpretation operator ${I}$. Theorem \ref{THM:main result} establishes that in NSC there is a testable condition that can be applied to any given choice data to determine whether this data can be thought of as resulting from an AIC procedure. 

Now, suppose we have choice data that is consistent with the AIC logic. The question we address next is about the extent to which the two key parameters of the AIC procedure can be uniquely identified from such data. We first consider the question of identification of the AI's preferences. In contrast to rational choice theory, with theories of bounded rationality like the AIC, there may be multiple possible preferences which can rationalize the same choice data. To check whether the AI ranks $x$ above $y$, it, thus, seems natural to check whether every possible representation of choices ranks $x$ above $y$. The following definition is useful to organize the discussion.

\begin{definition}
Let $c$ be an AIC. We say that $x$ is revealed to be preferred to $y$ by the AI, if for any $(\succ,{I})$ that is part of a AIC representation of $c$, we have $x\succ y$.
\end{definition}

Checking every possible representation of choices may not be a very practicable method. Fortunately, there is a much simpler way to identify the AI's preferences. It heavily draws on the idea of combining choice reversals and subset relationships which is underlying NSC. Indeed, we define the following binary relation $\succ$ on $X$ via
$$
x\succ y\text{ if }c(T)=x\text{ and }c(S)=y\text{ for some }S,T\in\mathcal{P}(X)\text{ with }S\subset T
$$
The binary relation $\succ$ may naturally be incomplete, but never intransitive. What is more, by Lemma \ref{LEM:asymmetric}, $\succ$ is asymmetric and, by Lemma \ref{LEM:acyclic}, it is also acyclic. As such, it has a transitive closure which we denote by $\succ^*$. If $x\succ^*y$, then $x$ is revealed to be preferred to $y$. Loosely speaking, this is true because if $x\succ z$ and $z\succ y$, then the transitivity of the underlying strict preference relation allows us to infer that $x$ is indeed revealed to be preferred to $y$.  The question remains whether $\succ^*$ really captures all revealed preferences and, at the same time, not more than that. The next proposition establishes that $\succ^*$ really is the revealed preference.

\begin{proposition}\label{PROP:revealed preference}
Let $c$ be an AIC. Then $x$ is revealed to be ``preferred'' to $y$ if and only if $x\succ^*y$.
\end{proposition}
\begin{proof}
Please refer to Appendix \ref{PROOF:revealed preference}.
\end{proof}

This result is also illuminating with respect to the characterization of AIC. Indeed, it enables us to restate NSC as the demand that the AI's underlying preferences are acyclic. To see this, note that we can use the revealed preference $\succ^*$ to formally restate NSC in the following way
$$
\text{If }\,x_{n+1}\succ^*x_n,\; x_{n}\succ^*x_{n-1},\dots,\; x_2\succ^* x_1,\text{ then }\neg[x_1 \succ^* x_{n+1}]
$$

Proposition \ref{PROP:revealed preference} furthermore shows that multiple preferences may be consistent with AIC choice data. As such, all of the AI's recommendations may align with the DM's preferences without the underlying preferences truly aligning.

Next, we consider the question of identifying the AI's interpretation operator. Again, to check whether the AI considers $x$ at $S$ for its choice, it seems natural to check whether every possible AIC-representation of the AI's choices specifies that $x$ receives the AI's consideration at $S$. In a similar way as before, the following definition is useful to organize the discussion.

\begin{definition}
Let $c$ be a CSC. We say that $x$ is revealed to receive consideration at $S\in\mathcal{P}(X)$ by the AI, if for any $(P,{I})$ that is part of an AIC representation of $c$, we have $x\in{I}(S)$
\end{definition}

It turns out, there also exists a simple way to identify the AI's revealed consideration set. To this end, we define the following consideration set ${I}^*$ on $\mathcal{P}(X)$ via
$$
x\in{I}^*(T)\text{ if }c(S)=x\text{ for some }S\in\mathcal{P}(X)\text{ with }S\subset T
$$

Again, the question remains whether ${I}^*$ really captures all revealed consideration and, at the same time, not more than that. The next proposition establishes that ${I}^*$ really is the revealed interpretation operator.

\begin{proposition}\label{PROP:revealed consideration}
Let $c$ be an AIC. Then $x$ is revealed to the AI's consideration at $T$ if and only if $x\in{I}^*(T)$
\end{proposition}
\begin{proof}
Please refer to Appendix \ref{PROOF:revealed consideration}.
\end{proof}

\begin{example}
Let $X=\{x,y,z\}$ denote the set of alternatives and $x \succ y\succ z$ the DM's preferences. The AI's recommendations are summarized as follows
\begin{table}[h]
    \centering
    \begin{tabular}{c| c c c c}
        $S$ & $\{x,y\}$ & $\{y,z\}$ & $\{x,z\}$ & $\{x,y,z\}$\\
        \hline
        ${I}(S)$ & $\{x\}$ & $\{y\}$ & $\{x\}$ & $\{x,y\}$ \\
        \hline
        $c(S)$ & $x$ & $y$ & $x$ & $x$ 
    \end{tabular}
\end{table}

Observe that this choice data satisfies NSC. But, the preferences underlying the corresponding AIC are not uniquely identifiable. Indeed, any preference ranking $\succ$ on $X$ satisfying $x\succ y$ is a potential candidate. In other words, it is not clear whether the AI ranks $z$ above $x$ and $y$, in between or below.
\end{example}

\section{Rational AI Agent Choice}\label{SECT:Rational AI Agent Choice}
The monotonic interpretation operator of an AIC preserves the subset-relationship between the non-empty subsets of $X$. But, it does not prevent the AI from interpreting two sets $S,T\in\mathcal{P}(X)$ which are incomparable with respect to set inclusion ($S\nsubseteq T$ and $S\nsubseteq T$) into sets sharing a subset-relation (${I}(S)\subseteq {I}(T)$). 

To address this, we next strengthen the connection between the subset-relationship in $\mathcal{P}(X)$ and the one in $\{{I}(S)|S\in\mathcal{P}(X)\}$. The next definition pins down how one can do so and refines the AIC procedure accordingly.

\begin{definition}
A choice function $c:\mathcal{P}(X)\to X$ is a rational AI agent's choice (RAIC) if it is an AIC and the corresponding interpretation operator satisfies
\begin{itemize}[leftmargin=40pt]
\item[$(T_{DM})$] \textsc{Double monotonicity ---} For all $S,T\in\mathcal{P}(X):$ ${I}(S)\subseteq {I}(T)\Leftrightarrow S\subseteq T$
\end{itemize}
\end{definition}

Double Monotonicity does not only ensure that the interpretation operator preserves the subset-relationship between the non-empty subsets of $X$. Rather, it also matches every subset-relation in $\{{I}(S)|S\in\mathcal{P}(X)\}$ uniquely with the corresponding subset-relation in $X$ prior to interpretation. In other words, the posets $(\mathcal{P}(X),\subseteq)$ and $(\{{I}(S)|S\in\mathcal{P}(X)\},\subseteq)$ are order isomorphic.

Next, we turn to a behavioral characterization of the RAIC procedure. Three axioms characterize the rational version of the AIC procedure.


The first axiom rules out pairwise cycles. It is a standard and well-known condition in choice theory, here given in a slightly restated formulation.

\begin{axiom}[No Binary Cycles (NBC)]
For all $x,y,z\in X$:
$$
[c(\{x,y\})\neq c(\{y,z\})\text{ and }c(\{x,z\})\neq c(\{y,z\})]\Rightarrow [c(\{x,y\})= c(\{x,z\})]
$$
\end{axiom}

The formulation of NBC here restricts attention to which choices from binary sets differ and which have to coincide. This way, it avoids taking a stance on how exactly the AI interprets the corresponding choice problems.

The second axiom is a condition akin to Contraction Independence (also known as Property $\alpha$ or Chernoff's Condition). In our context, it rules out menu effects by requiring that if an alternative is chosen from a large set, it must also be chosen from any smaller subset whose interpretation contains it.

\begin{axiom}[C-Contraction Independence (CCI)]
For all $R,S,T\in\mathcal{P}(X)$:
$$
[R\subset S\subset T,\; c(R)=c(T)=x]\Rightarrow [c(S)=x]
$$
\end{axiom}

At the heart of CCI is the contraction from the set $T$ to its subset $S$. As for NSC the added layer of complexity comes again from the AI interpreting choice problems. The AI's choice from the set $R$ then establishes that the choice problem $S$ is
interpreted in such a way that the interpreted choice problem ${I}(S)$ involves $x$.

The third axiom specifies that a rational AI agent notices when two alternatives are different. This manifests in the agent's choice behavior as follows.

\begin{axiom}[Noticeable Difference (ND)]
For all $x,y\in X$:
$$
[x\neq y]\Rightarrow [c(\{x\})\neq c(\{y\})]
$$
\end{axiom}

Whenever two alternatives are different, ND requires that the agent's choice from the corresponding interpreted singleton sets respects this difference. In other words, the AI's choice from two singleton sets only coincides when these are the same sets.

The following result then establishes that NBC, CCI and ND together form a behavioral characterization of RAIC.

\begin{theorem}\label{THM:main result2}
Let $X$ be a finite set of alternatives. A choice function $c$ on $X$ is an RAIC if and only if it satisfies NBC, CCI, and ND.
\end{theorem}
\begin{proof}
Please refer to Appendix \ref{PROOF:main result2}.
\end{proof}

Next, we analyze what can be inferred from choice data that is consistent with the RAIC logic. While the same exercise for the AIC procedure in Section \ref{SECT:Behavioral Foundation} only allowed to partially identify the AI's preferences and the interpretation operator, for the RAIC we can always fully identify these two key parameters entering into the AI's decision-making procedure. The following result spells this out formally.

\begin{proposition}\label{PROP:revealed preference2}
Let $c$ be an RAIC. Then, we can fully identify the AI's interpretation operator ${I}$ and its preference relation $\succ$ via defining, for any $T\in\mathcal{P}(X)$,
$$
{I}(T)=\{c(\{x\})|\; x\in T\}
$$
and
$$
x\succ y\text{ if }c(T)=x\text{ and }c(S)=y\text{ for some }S,T\in\mathcal{P}(X)\text{ with }S\subset T
$$
\end{proposition}
\begin{proof}
Please refer to Appendix \ref{PROOF:revealed preference2}.
\end{proof}

An agent is rational in the traditional economic sense if the agent's choice function satisfies the weak axiom of revealed preferences (WARP). WARP demands that if an alternative $x$ is chosen from some set where $y$ is available, then $y$ is never chosen from any set where $x$ is available. As the behavioral characterization in Theorem \ref{THM:main result2} and the identification in Proposition \ref{PROP:revealed preference2} establish, there is a sense in which the AI's choice under a interpretation operator satisfying Double Monotonicity is rational. But, this does not mean that the AI's choice is necessarily grounded. That is, under a RAIC, it is still possible that $c(S)\notin S$, for some $S\in\mathcal{P}(X)$. The next example illustrates this.

\begin{example}
Let $X=\{x,y,z\}$, $x \succ y\succ z$, and ${I}$ be defined by
$$
{I}(\{x\})=\{y\},\; {I}(\{y\})=\{z\},\; {I}(\{z\})=\{x\},\; \text{and } {I}(S)=\bigcup_{w\in S} {I}(\{w\}),\text{if }S\in\mathcal{P}(X),\;|S|>1
$$
Then, the AI's choice is not always grounded. That is, some of its choices satisfy $c(S)\notin S$.
\begin{table}[h]
    \centering
    \begin{tabular}{c|c c c c c c c}
        $S$ & $\{x\}$ & $\{y\}$ & $\{z\}$ & $\{x,y\}$ & $\{y,z\}$ & $\{x,z\}$ & $\{x,y,z\}$\\
        \hline
        ${I}(S)$ & $\{y\}$ & $\{z\}$ & $\{x\}$ & $\{y,z\}$ & $\{x,z\}$ & $\{x,y\}$ & $\{x,y,z\}$\\
        \hline
        $c(S)$ & $y$ & $z$ & $x$ & $y$ & $x$ & $x$ & $x$\\
        \hline
        $c(S)\in S$ & \ding{55} & \ding{55} & \ding{55} & \checkmark & \ding{55} & \checkmark & \checkmark  
    \end{tabular}
\end{table}

Furthermore, the RAIC produces here choices which satisfy WARP on the level of $\{{I}(S)|S\in\mathcal{P}(X)\}$, but not on the level of $\mathcal{P}(X)$. In other words, the AI chooses as if it maximizes the preferences $\succ$ on the level of $\{{I}(S)|S\in\mathcal{P}(X)\}$.
\end{example}

\section{Grounded and Rational AI Agent Choice}\label{SECT:Grounded and Rational AI Agent Choice}
This section aims to ground a rational AI agent's choice. One way to do so is to prevent the interpretation operator from running into loops, where the AI keeps interpreting already interpreted choice problems. Formally, such looping can be avoided by requiring that the interpretation operator satisfies idempotence.

\begin{definition}
A choice function $c:\mathcal{P}(X)\to X$ is an AI agent's choice (GRAIC) if it is an RAIC and the corresponding interpretation operator $I$ satisfies
\begin{itemize}[leftmargin=35pt]
\item[$(I_{IP})$] \textsc{Idempotence ---} For all $S\in\mathcal{P}(X):$ ${I}({I}(S))= {I}(S)$
\end{itemize}
\end{definition}

Next, we turn to a behavioral characterization of the RAIC procedure. True to the aim of this section, standard WARP characterizes GRAIC.

\begin{definition}
A choice function $c:\mathcal{P}(X)\to X$ satisfies WARP if, for all $S,T\in\mathcal{P}(X)$:
$$
[c(S)=x,\; y\in S,\; x\in T]\Rightarrow [c(T)\neq y]
$$
\end{definition}

The next result then provides a behavioral characterization of GRAIC. It shows that choice data is consistent with the GRAIC procedure if and only if it satisfies WARP, i.e., is rational in the traditional economic sense.

\begin{theorem}\label{THM:main result3}
Let $X$ be a finite set of alternatives. A choice function $c$ on $X$ is an GRAIC if and only if it satisfies WARP. 
\end{theorem}
\begin{proof}
Please refer to Appendix \ref{PROOF:main result3}.
\end{proof}

Since a GRAIC is a RAIC, choice data consistent with the GRAIC logic allows to fully identify the AI's preferences and the interpretation operator. The next result shows this formally.

\begin{proposition}\label{PROP:revealed preference3}
Let $c$ be an RAIC. Then, we can fully identify the AI's interpretation operator ${I}$ and its preference relation $\succ$ via defining, for any $T\in\mathcal{P}(X)$,
$$
{I}(T)=T
$$
and
$$
x\succ y\text{ if }c(T)=x\text{ and }y\in T\text{ for some }T\in\mathcal{P}(X)
$$
\end{proposition}
\begin{proof}
Please refer to Appendix \ref{PROOF:revealed preference3}.
\end{proof}

\section{Elaborations and Extensions}\label{SECT:Elaborations and Extensions}
This section elaborates on the consistency of the interpretation operator for AIC and RAIC. It furthermore provides and analysis and characterization of choice behavior by a grounded but potentially irrational AI agent.

\subsection{Monotonic Interpretations}
There are two equivalent axiomatizations of the class of interpretations we consider for the AIC in Section \ref{SECT:AI Agent}. These axiomatizations are based on the following properties: 
\begin{itemize}
\item[$(I_{\cup})$] \textsc{Union closure ---} For all $S,T\in\mathcal{P}(X):$ ${I}(S)\cup {I}(T)\subseteq {I}(S\cup T)$
\item[$(I_{\cap})$] \textsc{Intersection closure ---} For all $S,T\in\mathcal{P}(X):$ ${I}(S\cap T)\subseteq {I}(S)\cap {I}(T)$
\end{itemize}

While Monotonicity preserves the subset relationship under the interpretation, Union and Intersection closure are defined entirely on the level of interpretations. Nevertheless, the following result shows the equivalence between the different axiomatizations of interpretations.

\begin{proposition}[\textsc{Equivalent axiomatizations}]\label{PROP:properties of intepretation}
The following are equivalent:
\begin{itemize}
\item[(i)] An operator $I$ satisfies Monotonicity $(I_M)$.
\item[(ii)] An operator $I$ satisfies Union closure $(I_{\cup})$.
\item[(iii)] An operator $I$ satisfies Intersection closure $(T_{\cap})$.
\end{itemize}
\end{proposition}
\begin{proof}
Please refer to Appendix \ref{PROOF:properties of tranformer}.
\end{proof}

There are two ways to read the previous result. First, if $(I_{\cup})$ or $(I_{\cap})$ seem more appealing than $(I_M)$, we can simply replace the latter with one of the former in Definition \ref{DEF:Interpretation operator}. Second, regardless which of the three properties you treat as an axiom, the other two are going to follow as results.

\subsection{Doubly Monotonic Interpretations}
This section shows that interpretations satisfying Double Monotonicity are logically consistent. That is, the interpretation respects that complements cannot overlap. The following result characterizes interpretation operators that are consistent in a number of different ways.

\begin{theorem}[\textsc{Logical Consistency}]\label{THM:consistent interpretation}
The next statements are equivalent:
\begin{itemize}
\item[(i)] The interpretation operator $I$ is consistent, i.e., 
\begin{itemize}
\item[] For all $S\in\mathcal{P}(X)$: $I(S)\cap I(S^c)=\emptyset$.
\end{itemize}
\item[(ii)] The interpretation operator satisfies Double Intersection Closure, i.e., we have:
\begin{itemize}
\item[] For all $S,T\in\mathcal{P}(X):$ $I(S\cap T)=I(S)\cap I(T)$.
\end{itemize}
\item[(iii)] The interpretation operator satisfies Double Monotonicity, i.e., we have:
\begin{itemize}
\item[] For all $S,T\in\mathcal{P}(X):$ $S\subseteq T \Leftrightarrow I(S)\subseteq I(T)$.
\end{itemize}
\item[(iv)] The interpretation operator $I$ satisfies Negation-Elimination, i.e., we have:
\begin{itemize}
\item[] For all $S\in\mathcal{P}(X):$ $I(S^c)= (I(S))^c$.
\end{itemize}
\item[(v)] The interpretation operator $I$ is injective, i.e., we have:
\begin{itemize}
\item[] For all $S,T\in\mathcal{P}(X):$ $S\neq T \Rightarrow I(S)\neq I(T)$.
\end{itemize}
\item[(vi)] The interpretation operator $I$ is surjective, i.e., we have:
\begin{itemize}
\item[] For all $T\in\mathcal{P}(X)$ there exists $S\in\mathcal{P}(X)$ such that $I(S)=T$.
\end{itemize}
\item[(vii)] The interpretation operator $I$ is a relabelling, i.e., it satisfies:
\begin{itemize}
\item[(a)] There exists an automorphism $\rho:{X}\rightarrow{X}$ such that, for all $x\in X$: $I(\{x\})=\{\rho(x)\}$ 
\item[(b)] It satisfies Double Union Closure, i.e., for all $S,T\in\mathcal{P}(X):$ $I(S\cup T)=I(S)\cup I(T)$.
\end{itemize}
\end{itemize}
\end{theorem}
\begin{proof}
Please refer to Appendix \ref{PROOF:consistent interpretation}.
\end{proof}

The previous result implies that there are different ways one can identify inconsistencies. According to part (ii), interpretation is consistent if and only if the AI agent does not make mistakes with conjunctions. According to part (iii), she is correct with all the logical implications. According to part (iv), she is correct with negations (or taking complements). According to part (v) she does not confuse two different sets for being the same. According to part (vi), she associates every choice problem in $\mathcal{P}(X)$ with some (not necessarily the same) choice problem in $\mathcal{P}(X)$.

The previous characterizations suggest that a logically consistent AI agent may still make mistakes in what different choice problems actually entail, but does not make mistakes in how different choice problems logically relate to each other. That is why ---as the last part of the theorem indicates--- the only misinterpretations that are consistent are those that simply relabel the singleton set in $\mathcal{P}(X)$. In such cases, all the rules of logic are satisfied.

\subsection{Grounding Irrational and Rational Interpretations}
Instead of first ensuring rationalizability and then grounding the AI's choice, this section starts by grounding the choice first. To do so, it defines a grounded interpretation operator as follows.

\begin{definition}\label{DEF:Transformer2}\textsc{(Grounded Interpretation).}
A grounded interpretation is a mapping ${I}:\mathcal{P}(X)\rightarrow \mathcal{P}(X)$ that satisfies the following two basic conditions:
\begin{itemize}[leftmargin=45pt]
\item[$(I_{CON})$] \textsc{Consistency ---} For all $S,T\in\mathcal{P}(X)$, $S\subseteq T$: ${I}(S)\cap{I}(T^c)=\emptyset$
\item[$(I_{SIP})$] \textsc{Singleton Idempotence ---} For all $x\in X:$ ${I}({I}(\{x\}))={I}(\{x\})$
\end{itemize}
\end{definition}

Next, we introduce the choice procedure of an AI with grounded interpretation.

\begin{definition}
A choice function $c:\mathcal{P}(X)\to X$ is a grounded AI agent's choice (GAIC) if there exists a strict preference relation $\succ$ on $X$ and a grounded interpretation ${I}$ on $\mathcal{P}(X)$ such that for any choice problem $S\in\mathcal{P}(X)$,
$$
c(S)=\{x\in {I}(S)|\; x\succ y,\text{ for all }y\in {I}(S)\setminus\{x\}\}
$$
\end{definition}

As it turns out, this procedure can be characterized by the following single axiom.

\begin{axiom}[Groundedness]
For all $S\in \mathcal{P}(X)$:
$$
c(S)\in S
$$
\end{axiom}

The following result then establishes that Groundedness alone forms a behavioral characterization of GAIC.

\begin{theorem}\label{THM:main result4}
Let $X$ be a finite set of alternatives. A choice function $c$ on $X$ is a GAIC if and only if it satisfies Groundedness. 
\end{theorem}
\begin{proof}
Please refer to Appendix \ref{PROOF:main result4}.
\end{proof}

As in previous sections, we next analyze what can be inferred from choice data that is consistent with the GAIC logic. As it turns out, this procedure only allows to identify the absolute minimum, leaving the preferences completely unidentifiable. The following result spells this out formally.

\begin{proposition}\label{PROP:revealed preference4}
Let $c$ be an GAIC. Then, for the grounded distortion ${I}$, we can identify that, for any $T\in\mathcal{P}(X)$, it holds that
$$
\{c(T)\}\subseteq {I}(T)\subseteq T
$$
whereas there is nothing that can be identified about the preference relation $\succ$.
\end{proposition}
\begin{proof}
Please refer to Appendix \ref{PROOF:revealed preference4}.
\end{proof}

Next, we impose monotonicity as an additional property on the interpretation operator. The reason being that this now ensures that we can fully identify the AI's preferences.

\begin{definition}
A choice function $c:\mathcal{P}(X)\to X$ is a grounded \& monotonic AI agent's choice (GMAIC) if it is an GAIC and the corresponding interpretation $I$ satisfies
\begin{itemize}[leftmargin=30pt]
\item[$(T_M)$] \textsc{Monotonicity ---} For all $S,T\in \mathcal{P}(X):$ $S\subseteq T \Rightarrow {I}(S)\subseteq {I}(T)$
\end{itemize}
\end{definition}

Indeed, the following result shows that WARP characterizes monotonic version of the GAIC procedure.

\begin{theorem}\label{THM:main result5}
Let $X$ be a finite set of alternatives. A choice function $c$ on $X$ is a GMAIC if and only if it satisfies WARP. 
\end{theorem}
\begin{proof}
Please refer to Appendix \ref{PROOF:main result5}.
\end{proof}

As intended, choice data consistent with the GMAIC logic allows to fully identify the AI's preferences. In addition, also the interpretation is fully identifiable as the next result shows.

\begin{proposition}\label{PROP:revealed preference5}
Let $c$ be an GMAIC. Then, we can fully identify the AI's grounded interpretation $\mathcal{I}$ and its preference relation $\succ$ via defining, for any $T\in\mathcal{P}(X)$,
$$
{I}(T)=T
$$
and
$$
x\succ y\text{ if }c(T)=x\text{ and }y\in T\text{ for some }T\in\mathcal{P}(X)
$$
\end{proposition}
\begin{proof}
Please refer to Appendix \ref{PROOF:revealed preference5}.
\end{proof}

\section{Conclusion}\label{SECT:Conclusion}
This paper proposes and analyzes a model of choice for an agentic AI. We view the AI as a choice function that makes recommendations for each menu. The caveat is that the AI may misinterpret the menu from which the DM is seeking to choose. Our main characterization results show (i) under what conditions the AI is grounded in the sense of only recommending alternatives which are feasible choices, (ii) when we can ensure that the AI's recommendations align with the DM's preferences, and (iii) how we can guarantee that the AI makes recommendations which are internally consistent or rational in the traditional economic sense.

Our framework here suggests several avenues for future research. First, while our model is deterministic, a natural extension is in the direction of stochastic choice. In this case, choice data does not come in the form of a single alternative per choice problem. Instead, an analyst is then able to observe a distribution, representing the shares of a finite number of choices. Second, our identification results provide a theoretical benchmark for empirical ``revealed preference'' tests on AI agents. Via such tests we can then inspect the latent aspects of these systems, making sure they align with the preferences of the individual on whose behalf the AI agent makes choices.

\clearpage
\appendix
\section{Proofs}
\subsection{Proof of Theorem \ref{THM:main result}}\label{PROOF:main result}

\begin{proof}
\underline{Necessity:} NSC. For $i=1,\dots,n$, let $S_i, T_i\in\mathcal{P}(X)$ be such that $S_i\subset T_i$, $c(S_i)=x_i$, and $c(T_i)=x_{i+1}$. Further, let $S_{n+1},T_{n+1}\in\mathcal{P}(X)$ be such that $c(S_{n+1})=x_{n+1}$ and $S_{n+1}\subset T_{n+1}$. For NSC to hold, we have to show that this implies that $c(T_{n+1})\neq x_1$. To this end, note that, for all $i=1,\dots, n+1$, $c(S_i)=x_i$ implies that $x_i\in {I}(S_i)$. From monotonicity, it then follows that also $x_i\in{I}(T_i)$, for all $i=1,\dots, n+1$,. Furthermore, for all $i=1,\dots, n$, $c(T_i)=x_{i+1}$ implies that $x_{i+1}\in {I}(T_i)$. By definition of AIC, choice from $T_1$ then implies that $x_2\succ x_1$, choice from $T_2$ then implies that $x_3\succ x_2$,\dots, and choice from $T_n$ implies that $x_{n+1}\succ x_n$. It follows that $x_{n+1}\succ x_n\succ\dots\succ x_1$ and, in particular, by transitivity of $\succ$, that $x_{n+1}\succ x_1$. As such, $x_{n+1}\in{I}(T_{n+1})$ implies that any AIC satisfies $c(T_{n+1})\neq x_1$

\underline{Sufficiency:} Suppose that $c$ satisfies NSC. We construct the interpretation operator ${I}:\mathcal{P}(X)\to\mathcal{P}(X)$ and the preference relation $\succ$ explicitly.

First, define, for any $T\in\mathcal{P}(X)$, that
$$
x\in{I}(T)\quad\Leftrightarrow\quad c(S)=x,\text{ for some }S\subset T
$$
Note that, by this definition ${I}$ is clearly monotonic. 

Next, for any $x\neq y$, define a binary relation $P\subseteq X\times X$ via
$$
xPy\quad\Leftrightarrow\quad\exists S\in\mathcal{P}(X)\text{ such that }c(S)=x\text{ and }y\in{I}(S)
$$

\begin{lemma}\label{LEM:asymmetric}
$P$ on $X$ as defined above is asymmetric.
\end{lemma}
\begin{proof}
Towards a proof by contradiction, suppose that $P$ is not asymmetric, i.e., there exist $x_1,x_2\in X$ such that $x_1Px_2$ and $x_2Px_1$. By definition of $P$ above, this means that there exist two sets $S_1,S_2\in\mathcal{P}(X)$ with $x_1,x_2\in {I}(S_1)$, $x_1,x_2\in {I}(S_2)$, and $c(S_1)=x_1$, $c(S_2)=x_2$. By definition of ${I}$ above, it follows that there also exist two sets $T_1,T_2\in\mathcal{P}(X)$ such that $S_1\subset T_1$, $S_2\subset T_2$, and $c(T_1)=x_2$, $c(T_2)=x_1$. Observe, however, that the fact that $c(S_1)=x_1$, $c(T_1)=x_2$, $S_1\subset T_1$, $c(T_2)=x_2$, and $S_2\subset T_2$ imply that $c(T_2)\neq x_2$. Hence, we have arrived at our desired contradiction.
\end{proof}

\begin{lemma}\label{LEM:acyclic}
$P$ on $X$ as defined above is acyclic.
\end{lemma}
\begin{proof}
Towards a proof by contradiction, suppose that $P$ is not acyclic, i.e., there exist $n\in\mathbb{N}_{\geq1}$ and $x_1,\dots, x_{n+1}\in X$ such that $x_{n+1}Px_n$, $\dots$, $x_2Px_1$, $x_1Px_{n+1}$. By definition of $P$ above, this means that there exist sets $S_1,\dots, S_{n+1}\in\mathcal{P}(X)$ with $x_1,x_2\in {I}(S_1)$, $\dots$, $x_n,x_{n+1}\in {I}(S_n)$, $x_{n+1},x_1\in {I}(S_{n+1})$, and $c(S_i)=x_i$, for all $i=1,\dots, n+1$. By definition of ${I}$ above, it follows that there also exist sets $T_1,\dots, T_{n+1}\in\mathcal{P}(X)$ such that $S_i\subset T_i$, for all $i=1,\dots, n+1$, and $c(T_1)=x_2$, $\dots$, $c(T_n)=x_{n+1}$, $c(T_{n+1})=x_{1}$. This is a clear violation of NSC. Hence, we have arrived at our desired contradiction.
\end{proof}

We now verify that $({I},P)$ represent $c$ on $X$ as an AIC. To that end, pick any $T\in\mathcal{P}(X)$. Let $c(T)=x$. By our definition of ${I}$ above, it follows that $x\in{I}(T)$. Next, we show that there is no alternative $y\in{I}(T)$ with $yPx$. Suppose to the contrary that there is such a $y\in{I}(T)$. By our definition of ${I}$ above, it follows that there is an $S\in\mathcal{P}(X)$ such that $S\subset T$ and $c(S)=y$. On the other hand, $yPx$ implies that there is $S', T'\in\mathcal{P}(X)$ such that $S'\subset T'$ and $c(S')=x$, $c(T')=y$. But then, $c(S')=x$, $c(T')=y$, $S'\subset T'$, and $c(S)=y$. So, NSC implies that $c(T)\neq x$, for all $T\in\mathcal{P}(X)$ with $S\subset T$, and we have arrived at our desired contradiction. It follows that $\neg[zPx]$, for all $z\in{I}(T)$. Furthermore, choice from $T$ is decisive, since, by our definition of $P$ above, it holds that $xPz$, for all $z\in{I}(T)$, $z\neq x$. This establishes our desired conclusion.
\end{proof}

\subsection{Proof of Proposition \ref{PROP:revealed preference}}\label{PROOF:revealed preference}
\begin{proof}
Necessity: Take any pair of $x$ and $y$ without $x\succ^*y$. Then there exists a strict preference $\succ'$ including $\succ^*$ such that $y\succ'x$. By the proof of Theorem \ref{THM:main result}, there exists a interpretation operator ${I}$ such that $({I},\succ')$ represents $c$. Since $y\succ'x$, by definition, $x$ cannot be revealed to be preferred to $y$.

Sufficiency: We have already shown in Section \ref{SECT:Behavioral Foundation} that if $x\succ y$, then $x$ is revealed to be preferred to $y$. Now, consider the case when $x\succ^*y$. Since $\succ^*$ is defined as the transitive closure of $\succ$, this implies that there exists a sequence $(z_m)m=1^M$ in $X$ such that $x\succ z_1$, $z_1\succ z_2$, $\dots$, $z_M\succ y$. In this case, we know that for any $\succ^*$ that is part of the representation, $\succ\subseteq \succ^*$ and, hence $x\succ^* z_1$, $z_1\succ^* z_2$, $\dots$, $z_M\succ^* y$. Further, since $\succ^*$ is transitive it follows that $x\succ^*y$ and, thus, $x$ is revealed to be preferred to $y$

\end{proof}

\subsection{Proof of Proposition \ref{PROP:revealed consideration}}\label{PROOF:revealed consideration}
\begin{proof}
Necessity: Take any pair of $S\in\mathcal{P}(X)$ and any $x\in{I}(S)$ with $x\notin{I}^*(S)$. By the proof of Theorem \ref{THM:main result}, there exists a interpretation operator ${I}$ with $x\notin{I}(S)$ and a ranking $P$ such that $(P,{I})$ represents $c$. Since $x\notin{I}(S)$, by definition, $x$ cannot be revealed to receive consideration at $S$.

Sufficiency: Let $x\in{I}^*(T)$. Then, $c(S)=x$, for some $S\in\mathcal{P}(X)$ with $S\subset T$. Clearly, $c(S)=x$ implies that $x\in{I}(S)$. Monotonicity implies that also $x\in{I}(T)$. Then, it follows from the definition of the interpretation operator that $x\in{I}(T)$ for any ${I}$ that is part of a AIC representation of these choices.
\end{proof}

\subsection{Proof of Theorem \ref{THM:main result2}}\label{PROOF:main result2}
\begin{proof}
\underline{Necessity:} First, note that Double Monotonicity implies that, for all $S\in\mathcal{P}(X)$, it holds that $|T|=|{I}(T)|$. To see this, suppose otherwise. That is, there exists $T\in\mathcal{P}(X)$ such that $|T|\neq|{I}(T)|$. Note that then monotonicity implies that there exist $R,S\in\mathcal{P}(X)$ satisfying either $R\subsetneq S\subseteq T$, or, $T\subsetneq R\subseteq S$ such that $|{I}(R)|=|{I}(S)|$. Since $R\subset S$, monotonicity implies that ${I}(R)={I}(S)$. But, then $R\subsetneq S$ and ${I}(R)\subseteq {I}(S)$ together violate Double Monotonicity. This establishes that Double Monotonicity implies that, for all $S\in\mathcal{P}(X)$, it holds that $|S|=|{I}(S)|$.

NBC. Let $x,y,z\in X$ be such that $c(\{x,y\})\neq c(\{x,z\})\neq c(\{y,z\})$. For NBC to hold, we have to show that then $c(\{x,y\})= c(\{y,z\})$. To see this, first observe that, by the argument laid out above, Double Monotonicity implies that $|{I}(\{x,y\})|=|{I}(\{y,z\})|=|{I}(\{x,z\})|=2$. Next, observe that $\{x,y\}$ and $\{y,z\}$ are not subsets of one another, so Double Monotonicity implies that ${I}(\{x,y\})$ and ${I}(\{y,z\}\}$ are also not subsets of one another. But, then monotonicity implies that ${I}(\{x,y\}),{I}(\{y,z\}), and {I}(\{x,z\})$ are all three different 2-element subsets of ${I}(\{x,y,z\})$. Hence, $c(\{x,y\})\neq c(\{x,z\})\neq c(\{y,z\})$ implies that $c(\{x,y\})=c(\{y,z\})$.

CCI. Let $R,S,T\in\mathcal{P}(X)$ be such that $R\subset S\subset T$ and $c(R)=c(T)=x$. For CCI to hold, we have to show that then $c(S)=x$. First, note that $c(T)=x$ implies that $x\succ y$, for all $y\in{I}(T)\setminus\{x\}$. Next, since $R\subset T$, Double Monotonicity implies that ${I}(R)\subset {I}(T)$. As such, also $x\succ y$, for all $y\in{I}(R)\setminus\{x\}$. Since $c(R)=x$, we have $x\in{I}(R)$ and, by monotonicity also $x\in{I}(S)$. Hence, $c(S)=x$.

ND. Let $x,y\in X$ be such that $x\neq y$. Observe that since $\{x\}$ and $\{y\}$ are not subsets of one another, Double Monotonicity implies that also ${I}(\{x\})$ and ${I}(\{y\})$ are not subsets of one another. From above, we know that Double Monotonicity implies that, for all $S\in\mathcal{P}(X)$, it holds that $|S|=|{I}(S)|$. Hence, $|{I}(\{x\})|=|{I}(\{y\})|=1$ and, thus, $c(\{x\})\neq c(\{y\})$.

\underline{Sufficiency:} Suppose that $c$ satisfies NBC, C-WWARP and Noticeable Difference. We construct the interpretation operator ${I}:\mathcal{P}(X)\to\mathcal{P}(X)$ and the preference relation $\succ$ explicitly.

First, define, for any $T\in\mathcal{P}(X)$, that
$$
{I}(T)=\{c(\{x\})|\; x\in B\}
$$
Note how this ensures that, for any $x\in X$, the set ${I}(\{x\})$ is a singleton. Furthermore, Noticeable Difference implies that, whenever $x\neq y$ also ${I}(\{x\})\neq {I}(\{y\})$.

Next, observe that, for any $x\in X$, Noticeable Difference implies that there exists a unique $u_x\in X$ such that $c(\{u_x\})=x$. Therefore, define a binary relation $P\subseteq X\times X$ via
$$
xPy\quad\Leftrightarrow\quad c(\{u_x,u_y\})=x
$$

\begin{lemma}\label{LEM:complete and transitive}
$P$ on $X$ as defined above is complete and transitive.
\end{lemma}
\begin{proof}
Take any $x,y\in X$. By our argument above there exist unique $u_x,u_y\in X$ such that either $c(\{u_x,u_y\})=x$ or $c(\{u_x,u_y\})=y$. As such, $P$ is complete. It is straightforward to establish that NBC implies that $P$ is also transitive.
\end{proof}

Now, take any $T\in\mathcal{P}(X)$ and let $c(T)=x$. Towards a proof by contradiction, suppose that there exists $z\in{I}(T)$ such that $zPx$. Then, by definition of $P$ above, it follows that $c(\{u_x,u_z\})=z$. But, then $c(\{u_x\})=x$, $c(\{u_x,u_z\})=z$, and $c(T)=x$ together with $\{u_x\}\subset\{u_x,u_z\}\subset T$ violates C-WWARP. Indeed, since $P$ is complete and transitive, have instead $xPz$, for all $z\in{I}(T)\setminus\{x\}$.
\end{proof}

\subsection{Proof of Proposition \ref{PROP:revealed preference2}}\label{PROOF:revealed preference2}
\begin{proof}
First, by the proof of Theorem \ref{THM:main result2}, Double Monotonicity implies that, for all $S\in\mathcal{P}(X)$, it holds that $|S|=|{I}(S)|$. Furthermore, from Noticeable Difference it follows that $x\neq y$ implies $c(\{x\})\neq c(\{y\})$. As such, defining, for any $T\in\mathcal{P}(X)$,
$$
{I}(T)=\{c(\{x\})|\; x\in T\}
$$
it holds that $|T|=|{I}(T)|$. Therefore, ${I}$ truly is the AI's unique interpretation operator recovered from choice data $c$.

Observe that the above reasoning implies that, for any $x,y\in X$, there exists $u_x,u_y\in X$ such that ${I}(\{u_x\})=\{x\}$, ${I}(\{u_y\})=\{y\}$, and ${I}(\{u_x,u_y\})=\{x,y\}$. As such, $c(\{u_x,u_y\})$ reveals whether the AI's preferences satisfy $x\succ y$ or $y\succ x$. Since ${I}(\{u_x\})=\{x\}$ and ${I}(\{u_y\})=\{y\}$, defining 
$$
x\succ y\text{ if }c(T)=x\text{ and }c(S)=y\text{ for some }S,T\in\mathcal{P}(X)\text{ with }S\subset T
$$
truly allows to fully identify the AI's preferences.
\end{proof}

\subsection{Proof of Theorem \ref{THM:main result3}}\label{PROOF:main result3}
\begin{proof}
\underline{Necessity:} As the proof of Theorem \ref{THM:main result2} lays out, Double Monotonicity implies that, for all $S\in\mathcal{P}(X)$, it holds that $|S|=|{I}(S)|$. Furthermore, for any two $S,T\in\mathcal{P}(X)$ with $S\neq T$, it holds that ${I}(S)\neq{I}(T)$. As such, Idempotence then implies that ${I}$ is the identity. That is, for all $S\in\mathcal{P}(X)$, it holds that ${I}(S)=S$.


WARP. Let $S,T\in\mathcal{P}(X)$ be such that $c(S)=x$, $y\in S$, and $x\in T$. For WARP to hold, we have to show that then $c(T)\neq y$. From above, it follows that ${I}(S)=S$ and ${I}(T)=T$. As such, $c(S)=x$ and $y\in S$ implies that $x\succ y$. Furthermore, $x\in T$ implies that $x\in{I}(T)$. Hence, $x\succ y$ implies that, by definition of an AIC, $c(T)\neq y$.

\underline{Sufficiency:} Suppose that $c$ satisfies WARP. 
We construct the interpretation operator ${I}:\mathcal{P}(X)\to\mathcal{P}(X)$ and the preference relation $\succ$ explicitly.

First, define, for any $T\in\mathcal{P}(X)$, that ${I}(T)=T$. Clearly, by this definition, ${I}$ satisfies Double Monotonicity and Idempotence.

Next, define for all $x,y\in X$ with $x\neq y$ that
$$
x\succ y\quad\Leftrightarrow\quad c(\{x,y\})=x
$$
Clearly, $\succ$ as thus defined is a complete and transitive binary relation.

Now, take any $T\in\mathcal{P}(X)$ and let $c(T)=x$. Towards a proof by contradiction, suppose that there exists $z\in{I}(T)=T$ such that $z\succ x$. Then, by definition of $\succ$ above, it follows that $c(\{x,z\})=z$. But, then $c(\{x,z\})=z$, $c(T)=x$, together with $z\in{I}(T)=T$ violates WARP. Therefore, we have arrived at our desired contradiction. As such, we can conclude that since $\succ$ is complete and transitive, we have instead $x\succ z$, for all $z\in{I}(T)\setminus\{x\}$. Hence, the tuple $(\succ,{I})$ rationalizes $c$.
\end{proof}

\subsection{Proof of Proposition \ref{PROP:revealed preference3}}\label{PROOF:revealed preference3}
\begin{proof}
First, by the proof of Theorem \ref{THM:main result3}, Double Monotonicity and Idempotence implies that, for all $S\in\mathcal{P}(X)$, it holds that ${I}(S)=S$. As such, defining, for any $T\in\mathcal{P}(X)$,
$$
{I}(T)=\{c(\{x\})|\; x\in T\}
$$
allows to truly recover the AI's unique interpretation operator from choice data $c$.

Furthermore, defining 
$$
x\succ y\text{ if }c(T)=x\text{ and }c(S)=y\text{ for some }S,T\in\mathcal{P}(X)\text{ with }S\subset T
$$
truly allows to fully identify the AI's complete and transitive preferences.
\end{proof}

\subsection{Proof of Proposition \ref{PROP:properties of intepretation}}\label{PROOF:properties of tranformer}

\begin{proof}
$(i)\Rightarrow(ii):$ Take any $S,T\in\mathcal{P}(X)$
\begin{align*}
S,T\in\mathcal{P}(X) &\Rightarrow {I}((S\cup T)\cap S)\subseteq {I}(S\cup T)\cap {I}(S)& \bigl(\mbox{by intersection closure}\bigr) \\
&\Rightarrow {I}(S)\subseteq {I}(S\cup T)\cap {I}(S)& \bigl(\mbox{by intersection of sets}\bigr) \\
&\Rightarrow {I}(S)\subseteq {I}(S\cup T) & \bigl(\mbox{by subset relation}\bigr)
\end{align*}
Analogously,
\begin{align*}
S,T\in\mathcal{P}(X) &\Rightarrow {I}((S\cup T)\cap T)\subseteq {I}(S\cup T)\cap {I}(T)& \bigl(\mbox{by intersection closure}\bigr) \\
&\Rightarrow {I}(T)\subseteq {I}(S\cup T)\cap {I}(T)& \bigl(\mbox{by intersection of sets}\bigr) \\
&\Rightarrow {I}(T)\subseteq {I}(S\cup T) & \bigl(\mbox{by subset relation}\bigr)
\end{align*}
Therefore, ${I}(S)\cup {I}(T)\subseteq {I}(S\cup T)$.

$(ii)\Rightarrow(iii):$ Take any $S,T\in\mathcal{P}(X)$
\begin{align*}
S\subseteq T &\Rightarrow {I}(S)\cup {I}(T\setminus S)\subseteq {I}(S\cup (T\setminus S))& \bigl(\mbox{by union closure}\bigr) \\
&\Rightarrow {I}(E)\cup {I}(T\setminus S)\subseteq {I}(T) & \bigl(\mbox{by union of sets}\bigr)\\
&\Rightarrow {I}(S)\subseteq {I}(T) & \bigl(\mbox{by subset relation}\bigr)
\end{align*}

$(iii)\Rightarrow(i):$ Take any $S,T\in\mathcal{P}(X)$
\begin{align*}
S,T\subseteq{X} &\Rightarrow S\cap T\subseteq S\text{ and }S\cap T\subseteq T& \bigl(\mbox{by intersection of sets}\bigr) \\
&\Rightarrow {I}(S\cap T)\subseteq {I}(S)\text{ and }{I}(S\cap T)\subseteq {I}(T) & \bigl(\mbox{by Monotonicity}\bigr)
\end{align*}
\end{proof}

\subsection{Proof of Theorem \ref{THM:consistent interpretation}}\label{PROOF:consistent interpretation}
\begin{proof}
Before starting with the proof, notice that $(ii)$ is equivalent to the weaker condition that states that $I(S)\subseteq I(T)$ implies $S\subseteq T$, as the converse already holds by monotonicity (Proposition \ref{PROP:properties of intepretation}). 

$(i)\Rightarrow(ii):$ Suppose that $I$ is consistent. We proceed to prove $(ii)$ by contraposition, i.e., we will show that $S\nsubseteq T$ implies $I(S)\nsubseteq I(T)$. If $S\cap T=\emptyset$, then by Proposition \ref{PROP:properties of intepretation}, we have $I(S)\cap I(T)=\emptyset$, and a fortiori $I(S)\nsubseteq I(T)$. If $S\cap T\neq\emptyset$, by Proposition \ref{PROP:properties of intepretation}, we have $I(S\setminus T)\cap I(T)=\emptyset$, meaning $I(S\setminus T)\nsubseteq I(T)$. The latter combined with $I(S\setminus T)\subseteq I(S)$, implies $I(S)\nsubseteq I(T)$, which completes this part of the proof.

$(ii)\Rightarrow(iii):$ By combining the condition in $(ii)$ with Proposition \ref{PROP:properties of intepretation}, we obtain 
\begin{equation}\label{EQ:Interpretation equivalence}
I(S)=I(T)\Leftrightarrow S=T
\end{equation}
Take a sequence of the states $({x}_1,\dots,{x}_K)$, so that each state appears exactly once. By definition
$$I(\{{x}_1\})\subseteq\ldots\subseteq I(\{{x}_1,\dots,{x}_k\})\subseteq\ldots\subseteq I(\{{x}_1,\dots,{x}_K\})$$
Suppose that at least one of these inclusions is not strict. Then, by Equivalence (\ref{EQ:Interpretation equivalence}), we have $\{{x}_1,\dots,{x}_k\}=\{{x}_1,\dots,{x}_{k+1}\}$, which is an obvious contradiction. Hence, we get
$$\underbrace{|I(\{{x}_1\})|}_{1}<\ldots< \underbrace{|I(\{{x}_1,\dots,{x}_k\})|}_{k}<\ldots<\underbrace{|I(\{{x}_1,\dots,{x}_K\})|}_{K}$$
meaning that $|I(S)|=|S|$ for every event $S\subseteq {X}$. So, we have $|I(\{{x}\})|=1$ for every ${x}\in{X}$. Furthermore, any two distinct states have a different image under $I$, by Equivalence (\ref{EQ:Interpretation equivalence}). Therefore $I$ is an automorphism when applied to the singletons. Finally, observe that it can only be that $I$ satisfies the condition in $(iii)$.

$(iii) \Rightarrow (iv):$ By the condition in (iii) it follows directly that
$$ (I(S))^c=\left(\bigcup_{{x}\in S}I(\{{x}\})\right)^c=\bigcup_{{x}\in  S^c}I(\{{x}\})=I(S^c)$$

$(iv) \Rightarrow (v):$ Take any $S,T\subseteq{X}$ such that $S\neq T$. Then, it either holds that $S\setminus T\neq\emptyset$, or, that $T\setminus S\neq\emptyset$. Observe that WLOG we can assume that it is $S\setminus T\neq\emptyset$. Then, there exists ${x}\in S$ such that ${x}\notin T$. Indeed, for such ${x}$ it holds that ${x}\in T^c$. By negation-elimination, $I(T^c)=\neg I(T)$. As such, by Monotonicity, our ${x}\in S\setminus T$ satisfies $\emptyset\neq I(\{{x}\})=I(\{{x}\}\cap T^c)\subseteq I(\{{x}\})\cap I(\neg T)\subseteq I(T^c)=\neg I(T)$. It follows, in particular, that $I(\{{x}\})\nsubseteq I(T)$. But, since ${x}\in S$, it holds that $(S\cap\{{x}\})\neq\emptyset$ such that, by Monotonicity, it holds that $I(\emptyset)\neq I({x})=I(S\cap\{{x}\})\subseteq I(S)\cap I(\{{x}\})\subseteq I(S)$. In summary, $\emptyset\neq I(\{{x}\})$, $I(\{{x}\})\subseteq I(S)$, but $I(\{{x}\})\nsubseteq I(T)$. Hence, $I(S)\neq I(T)$.

$(v) \Rightarrow (vi):$ Since $I$ is injective, each element of the codomain $\mathcal{P}(X)$ is mapped to by at most one element of the domain $\mathcal{P}(X)$. Since domain and codomain coincide, this implies that each element of the codomain is mapped to by exactly one element of the domain. As such, $I$ is surjective. In other words, for every $T\subseteq{X}$, there exists $S\subseteq {X}$ such that $I(S)=T$.

$(vi) \Rightarrow (vii):$ First, we show that there exists an automorphism $\rho:{X}\to{X}$ such that $I(\{{x}\})=\{\rho({x})\}$, for all ${x}\in{X}$. To this end, take any ${x}\in{X}$. By $I$ being surjective, there exists $S\subseteq{X}$ such that $I(S)=\{{x}\}$. Observe that, for all ${y}\in S$, Monotonicity then imply that 
$$
\emptyset\neq I(\{{y}\})=I( S\cap \{{y}\})\subseteq I(S)\cap I(\{{y}\})\subseteq I(S)=\{{x}\}
$$
Since $\emptyset\neq I(\{{y}\})$, it follows that $I(\{{y}\})=\{{x}\}$. By $I$ being surjective and $I$'s codomain coinciding with its domain, it furthermore follows that $S=\{{y}\}$. The same argument then implies that there exists an automorphism $\rho:{X}\to{X}$ such that $I(\{{x}\})=\{\rho({x})\}$, for all ${x}\in{X}$.

Second, we show that Double Union Closure holds, i.e., that for every $S,T\subseteq{X}$: $I(S\cup T)=I(S)\cup I(T)$. By Proposition \ref{PROP:properties of intepretation}, it holds that $I(S)\cup I(T)\subseteq I(S\cup T)$. Therefore, it only remains to show that $I(S\cup T)\subseteq I(S)\cup I(T)$.

To show this, suppose to the contrary that there exists $S,T\subseteq{X}$ such that $I(S\cup T)\nsubseteq I(S)\cup I(T)$, i.e., $I(S\cup T)\setminus (I(S)\cup I(T))\neq\emptyset$. Now, take any ${x}\in \left(I(S\cup T)\setminus (I(S)\cup I(T))\right)$. Since $I$ is surjective, there exists $R\subseteq{X}$ such that $I(R)=\{{x}\}$. In fact, from our automorphism above, we know that $R=\{{y}\}$, for some ${y}\in{X}$. That is, $I(\{{y}\})=\{{x}\}$. By Monotonicity, ${y}\notin S\cup T$. To see this, suppose otherwise, i.e., ${y}\in S\cup T$. Observe that WLOG we can assume that ${y}\in S$. Then, by Monotonicity, 
$$
\{{x}\}=I(\{{y}\})=I(S\cap\{{y}\})\subseteq I(S)\cap I(\{{y}\})\subseteq I(S)
$$
And, thus, ${x}\in I(S)$, violating that ${x}\in \left(I(s\cup T)\setminus (I(S)\cup I(T))\right)$. This shows that ${y}\notin S\cup T$.

Therefore, $(\{{y}\}\cup S\cup T)\neq S\cup T$. Since $I$ is surjective, this implies that also $I(\{{y}\}\cup(S\cup T))\neq I(S\cup T)$. By Monotonicity, it furthermore holds that $I(\{{y}\}\cup(S\cup T))\supsetneq I(S\cup T)$. Iterating this further, implies that there exists ${x}\in{X}$ such that $(\{{x}\}\cup\{{y},\dots\}\cup S\cup T)={X}\neq (\{{y},\dots\}\cup S\cup T)={X}\setminus\{{x}\}$, but $I({X}\setminus\{{x}\})={X}$. Since, by Monotonicity, $I({X})\supseteq I({X}\setminus\{{x}\})$, this implies that $I({X})=I({X}\setminus\{{x}\})={X}$ which violates the fact that $I$ is surjective. Therefore, $I(S\cup T)\subseteq I(S)\cup I(T)$. Hence, $I(S\cup T)= I(S)\cup I(T)$

$(vii) \Rightarrow (i):$ Take any $E\subseteq{X}$ and consider $I(S)\cap I(S^c)$. By Double Union Closure, $I(S)=\bigcup_{{x}\in S}I(\{{x}\})$ and $I(S^c)=\bigcup_{{x}\in S^c}I(\{{x}\})$. Since there exists an automorphism  $\rho:X\to X$ such that $I(\{x\})=\{\rho(x)\}$, for all $x\in X$, it follows that $I(\{x\})=\{\rho(x)\}\neq \{\rho(y)\}=I(\{y\})$, for any $x,y\in X$ with $x\neq y$. Hence, $I(S)\cap I( S^c)=\emptyset$.

\end{proof}

\subsection{Proof of Theorem \ref{THM:main result4}}\label{PROOF:main result4}
\begin{proof}
\underline{Necessity:} Groundedness. Note that Consistency implies that, for any $x\in X$ and all $S\subseteq X\setminus\{x\}$, it holds that ${I}(\{x\})\cap{I}(S)=\emptyset$. Since ${I}$ maps each non-empty subset of $X$ to a non-empty subset of $X$, this implies that, for any $x\in X$, the set ${I}(\{x\})$ is a singleton set. Furthermore, $x\neq y$ implies that ${I}(\{x\})\neq {I}(\{y\})$. But, then idempotence implies that, for all $x\in X$, it holds that ${I}(\{x\})=\{x\}$. 

As such, it follows from Consistency that, for all $S\in\mathcal{P}(X)$, it holds that ${I}(S)\subseteq S$. Since $c(S)\in{I}(S)$, this implies that $c(S)\in S$.

\underline{Sufficiency:} Suppose that $c$ satisfies Groundedness. 
We construct the grounded interpretation ${I}:\mathcal{P}(X)\to\mathcal{P}(X)$ and the preference relation $\succ$ explicitly.

First, define, for any $T\in\mathcal{P}(X)$, that ${I}(T)=c(T)$. Next, let $\succ$ be any complete and transitive binary relation on $X$.

Now, take any $T\in\mathcal{P}(X)$ and let $c(T)=x$. Since $x\in {I}(T)$ and, in particular, ${I}(T)=\{x\}$, maximizing any complete and transitive binary relation complete and transitive binary relation produces $c(T)=x$. Hence, the tuple $(\succ,{I})$ rationalizes $c$.
\end{proof}

\subsection{Proof of Proposition \ref{PROP:revealed preference4}}\label{PROOF:revealed preference4}
\begin{proof}
By the proof of Theorem \ref{THM:main result4}, for any $T\in{I}(T)$, we can infer that $\{c(T)\}\subseteq{I}(T)\subseteq T$.

Furthermore, the proof of Theorem \ref{THM:main result4} establishes that there we cannot infer anything about the preferences $\succ$ which are part of a GAIC.
\end{proof}

\subsection{Proof of Theorem \ref{THM:main result5}}\label{PROOF:main result5}
\begin{proof}
\underline{Necessity:} By the poof of Theorem \ref{THM:main result4}, for all $x\in X$, it holds that ${I}(\{x\})=\{x\}$. Monotonicity then implies that, for any fixed $S\in\mathcal{P}(X)$ and all $x\in S$, it holds that $x\in{I}(S)$. By Consistency, we then have ${I}(S)=S$ for any such $S$. As such GMAIC maximizes a complete and transitive preference relation over non-empty subsets of $X$, so standard results imply that the corresponding $c$ satisfies WARP.

\underline{Sufficiency:}Suppose that $c$ satisfies WARP. 
We construct the grounded interpretation ${I}:\mathcal{P}(X)\to\mathcal{P}(X)$ and the preference relation $\succ$ explicitly.

First, define, for any $T\in\mathcal{P}(X)$, that ${I}(T)=T$. Next, define 
$$
x\succ y\text{ if }c(T)=x\text{ and }c(S)=y\text{ for some }S,T\in\mathcal{P}(X)\text{ with }S\subset T
$$

Since $c$ satisfies WARP, $\succ$ as thus defined is complete and transitive. Furthermore, if $c(T)=x$, then $c(\{y\})=y$, for any $y\in T$ implies that $x\succ y$ such that the choice from $T$ is decisive. Hence, the tuple $(\succ,{I})$ rationalizes $c$.
\end{proof}

\subsection{Proof of Proposition \ref{PROP:revealed preference5}}\label{PROOF:revealed preference5}
\begin{proof}
First, by the proof of Theorem \ref{THM:main result5}, it follows that, for all $S\in\mathcal{P}(X)$, it holds that ${I}(S)=S$. As such, defining, for any $T\in\mathcal{P}(X)$,
$$
{I}(T)=\{c(\{x\})|\; x\in T\}
$$
allows to truly recover the grounded \& monotonic interpretation $I$ from choice data $c$.

Furthermore, defining 
$$
x\succ y\text{ if }c(T)=x\text{ and }c(S)=y\text{ for some }S,T\in\mathcal{P}(X)\text{ with }S\subset T
$$
truly allows to fully identify the AI's complete and transitive preferences.
\end{proof}


\clearpage
\begin{thebibliography}{}

\bibitem[Armouti-Hansen and Kops, 2018]{Armouti-Hansen2018}
Armouti-Hansen, J. and Kops, C. (2018).
\newblock This or that? sequential rationalization of indecisive choice
  behavior.
\newblock {\em Theory and Decision}, 84(4):507--524.

\bibitem[Armouti-Hansen and Kops, 2024]{Kops2024}
Armouti-Hansen, J. and Kops, C. (2024).
\newblock Managing anticipation and reference-dependent choice.
\newblock {\em Journal of Mathematical Economics}, 112:102988.

\bibitem[Baigent and Gaertner, 1996]{Baigent1996}
Baigent, N. and Gaertner, W. (1996).
\newblock Never choose the uniquely largest a characterization.
\newblock {\em Economic Theory}, 8(2):239--249.

\bibitem[Borah and Kops, 2020]{Borah2020}
Borah, A. and Kops, C. (2020).
\newblock Choice via social influence.
\newblock Working paper, University of Heidelberg.

\bibitem[Caplin et~al., 2025]{Caplin2025}
Caplin, A., Martin, D., and Marx, P. (2025).
\newblock Modeling machine learning: A cognitive economic approach.
\newblock {\em Journal of Economic Theory}, 224:105970.

\bibitem[Cherepanov et~al., 2013]{Cherepanov2013}
Cherepanov, V., Feddersen, T., and Sandroni, A. (2013).
\newblock Rationalization.
\newblock {\em Theoretical Economics}, 8(3):775--800.

\bibitem[Cuhadaroglu, 2017]{Cuhadaroglu2017}
Cuhadaroglu, T. (2017).
\newblock Choosing on influence.
\newblock {\em Theoretical Economics}, 12(2):477--492.

\bibitem[Dutta and Horan, 2015]{Horan2015}
Dutta, R. and Horan, S. (2015).
\newblock Inferring rationales from choice: Identification for rational
  shortlist methods.
\newblock {\em American Economic Journal: Microeconomics}, 7(4):179--201.

\bibitem[Eliaz et~al., 2011]{Eliaz2011}
Eliaz, K., Richter, M., and Rubinstein, A. (2011).
\newblock Choosing the two finalists.
\newblock {\em Economic Theory}, 46(2):211--219.

\bibitem[Eliaz and Spiegler, 2011]{Eliaz2011b}
Eliaz, K. and Spiegler, R. (2011).
\newblock Consideration sets and competitive marketing.
\newblock {\em The Review of Economic Studies}, 78(1):235--262.

\bibitem[Horan, 2016]{Horan2016}
Horan, S. (2016).
\newblock A simple model of two-stage choice.
\newblock {\em Journal of Economic Theory}, 162:372--406.

\bibitem[Katz and George, 2019]{Katz2019}
Katz, J.~D. and George, D.~T. (2019).
\newblock Reclaiming magical incantation in graduate medical education.
\newblock {\em Clinical Rheumatology}, 39(3):703--707.

\bibitem[Kim et~al., 2024]{Kim2024}
Kim, J., Kovach, M., Lee, K.-M., Shin, E., and Tzavellas, H. (2024).
\newblock Learning to be homo economicus: Can an llm learn preferences from
  choice.

\bibitem[Kops, 2018]{Kops2018}
Kops, C. (2018).
\newblock {(F)}lexicographic shortlist method.
\newblock {\em Economic Theory}, 65(1):79--97.

\bibitem[Kops, 2022]{Kops2022}
Kops, C. (2022).
\newblock Cluster-shortlisted choice.
\newblock {\em Journal of Mathematical Economics}, 102:102726.

\bibitem[Lleras et~al., 2017]{Lleras2017}
Lleras, J.~S., Masatlioglu, Y., Nakajima, D., and Ozbay, E.~Y. (2017).
\newblock When more is less: Limited consideration.
\newblock {\em Journal of Economic Theory}, 170:70--85.

\bibitem[Manzini and Mariotti, 2007]{Manzini2007}
Manzini, P. and Mariotti, M. (2007).
\newblock Sequentially {R}ationalizable {C}hoice.
\newblock {\em American Economic Review}, 97(5):1824 -- 1839.

\bibitem[Manzini and Mariotti, 2012a]{Manzini2012a}
Manzini, P. and Mariotti, M. (2012a).
\newblock Categorize then choose: Boundedly rational choice and welfare.
\newblock {\em Journal of the European Economic Association}, 10(5):1141--1165.

\bibitem[Manzini and Mariotti, 2012b]{Manzini2012b}
Manzini, P. and Mariotti, M. (2012b).
\newblock Choice by lexicographic semiorders.
\newblock {\em Theoretical Economics}, 7(1):1--23.

\bibitem[Manzini et~al., 2025]{Manzini2025}
Manzini, P., Mariotti, M., and \"Ulk\"u, L. (2025).
\newblock Choice and opportunity costs.
\newblock {\em The Review of Economic Studies}, page rdaf101.

\bibitem[Masatlioglu et~al., 2012]{Masatlioglu2012}
Masatlioglu, Y., Nakajima, D., and Ozbay, E.~Y. (2012).
\newblock Revealed attention.
\newblock {\em American Economic Review}, 102(5):2183--2205.

\bibitem[Russell and Norvig, 1995]{Russell1995}
Russell, S.~J. and Norvig, P. (1995).
\newblock {\em Artificial intelligence: A Modern Approach}.
\newblock Upper Saddle River, New Jersey: Prentice Hall.

\bibitem[Salant and Rubinstein, 2008]{Salant2008}
Salant, Y. and Rubinstein, A. (2008).
\newblock (a, f): Choice with frames.
\newblock {\em The Review of Economic Studies}, 75(4):1287--1296.

\bibitem[Shi et~al., 2024]{Shi2024}
Shi, W., Min, S., Yasunaga, M., Seo, M., James, R., Lewis, M., Zettlemoyer, L.,
  and Yih, W.-t. (2024).
\newblock Replug: Retrieval-augmented black-box language models.
\newblock In {\em Proceedings of the 2024 Conference of the North American
  Chapter of the Association for Computational Linguistics: Human Language
  Technologies (Volume 1: Long Papers)}, pages 8371--8384.

\bibitem[Stango and Zinman, 2014]{Stango2014}
Stango, V. and Zinman, J. (2014).
\newblock Limited and varying consumer attention: Evidence from shocks to the
  salience of bank overdraft fees.
\newblock {\em The Review of Financial Studies}, 27(4):990--1030.

\bibitem[Tyson, 2015]{Tyson2015}
Tyson, C.~J. (2015).
\newblock Satisficing behavior with a secondary criterion.
\newblock {\em Social Choice and Welfare}, 44(3):639--661.

\bibitem[Vaswani et~al., 2017]{Vaswani2017}
Vaswani, A., Shazeer, N., Parmar, N., Uszkoreit, J., Jones, L., Gomez, A.~N.,
  Kaiser, L.~u., and Polosukhin, I. (2017).
\newblock Attention is all you need.
\newblock In Guyon, I., Luxburg, U.~V., Bengio, S., Wallach, H., Fergus, R.,
  Vishwanathan, S., and Garnett, R., editors, {\em Advances in Neural
  Information Processing Systems}, volume~30.

\end{thebibliography}
\end{document}